\documentclass[11pt]{article}

\usepackage{fullpage}

\usepackage{epigraph}
\usepackage{nicefrac}
\usepackage[margin=1in]{geometry}
\usepackage{relsize}

\usepackage{dsfont}

\usepackage{latexsym}

\usepackage[dvipsnames,usenames,table]{xcolor}
\usepackage{parskip}
\usepackage{color}
\usepackage{float}
\usepackage{bbm}
\floatstyle{boxed}
\newfloat{algorithm}{t}{lop}
\usepackage{tikz}
\usepackage{varwidth}

\usepackage{float}
\usepackage[titletoc,title]{appendix}
\usepackage[classfont=sanserif,langfont=roman,funcfont=italic]{complexity}
\usepackage{caption}
\usepackage{comment}
\usepackage{enumerate}
\usepackage{amsmath, amsthm, amssymb, amstext,  graphicx,amsopn} 
\usepackage{subfigure}

\usepackage{bbm}

\usepackage{booktabs} 
\DeclareMathOperator{\Ex}{\mathbb{E}}

\usepackage{epigraph}
\usepackage{nicefrac}
\usepackage{hyperref}
\hypersetup{colorlinks=true,linkcolor=black,citecolor=blue}

\definecolor{burgundy}{rgb}{0.5, 0.0, 0.13}
\definecolor{crimson}{rgb}{0.86, 0.08, 0.24}\usepackage{mathtools,extarrows}
\usepackage{url}
\newtheorem{theorem}{Theorem}
\newtheorem{definition}{Definition}

\newtheorem{proposition}{Proposition}

\newtheorem{corollary}{Corollary}

\newtheorem{lemma}{Lemma}

\allowdisplaybreaks

\allowdisplaybreaks

\renewcommand{\vec}[1]{\mathbf{#1}}

\usepackage{thmtools} 
\usepackage{thm-restate}
\usepackage[capitalise,nameinlink]{cleveref}

\crefname{figure}{Figure}{Figure}

\usepackage{mathrsfs}
\hypersetup{colorlinks=true,linkcolor=blue,citecolor=blue}
\definecolor{mygreen}{rgb}{0.0, 0.55, 0.0}
\definecolor{blue-violet}{rgb}{0.54, 0.17, 0.89}

\usepackage{tikz}
\usetikzlibrary{calc, graphs, graphs.standard, shapes, arrows, arrows.meta, positioning, decorations.pathreplacing, decorations.markings, decorations.pathmorphing, fit, matrix, patterns, shapes.misc, tikzmark}

\usepackage{amsfonts}

\newcommand{\eps}{\varepsilon}
\renewcommand{\E}{\mathbb{E}}

\newcommand{\explain}[1]{\tag{\textcolor{gray}{#1}}}

\definecolor{airforceblue}{rgb}{0.36, 0.54, 0.66}
\definecolor{darkblue}{rgb}{0.0, 0.0, 0.55}

\title{The Randomized Query Complexity of Finding a Tarski Fixed Point on the Boolean Hypercube\footnote{This research was supported by US National Science Foundation CAREER grant CCF-2238372. Alphabetical author order.}}

\author{Simina Br\^anzei\footnote{Purdue University. E-mail: simina.branzei@gmail.com.} \and Reed Phillips\footnote{Purdue University. E-mail: phill289@purdue.edu.} \and Nicholas Recker\footnote{Epic. E-mail: nrecker@umich.edu.}}

\date{\today} 

\begin{document}

\maketitle 

\begin{abstract}
The Knaster-Tarski theorem, also known as Tarski's theorem,  guarantees that every monotone function defined on a complete lattice has a fixed point. 
We analyze the query complexity of finding such a fixed point on the 
$k$-dimensional grid of side length $n$ under the 
$\leq$ relation. Specifically, there is an unknown monotone function $f: \{0,1,\ldots, n-1\}^k \to \{0,1,\ldots, n-1\}^k$ and  an algorithm must query a vertex $v$ to learn  $f(v)$.

A key special case of interest is the Boolean hypercube $\{0,1\}^k$, which is isomorphic to the power set lattice—the original setting of the Knaster-Tarski theorem. We prove a lower bound that characterizes the randomized and deterministic query complexity of the Tarski search problem on the Boolean hypercube as $\Theta(k)$. More generally, we give a randomized lower bound of $\Omega\left( k + \frac{k \cdot \log{n}}{\log{k}} \right)$ for the $k$-dimensional grid of side length $n$, which is asymptotically tight in high dimensions when 
$k$ is large relative to $n$. 
\end{abstract}

\section{Introduction}


The Knaster-Tarski theorem, also known as Tarski's theorem, guarantees that every monotone function $f : \mathcal{L} \to \mathcal{L}$ defined over a complete lattice $(\mathcal{L}, \leq )$ has a fixed point. 
Tarski proved the most general form of the theorem \cite{tarski1955lattice}:
\begin{quote}
    \emph{Let $(\mathcal{L}, \leq)$ be a complete lattice and let $f : \mathcal{L} \to \mathcal{L}$ be an order-preserving (monotone) function with respect to $\leq$. Then the set of fixed points of $f$ in $\mathcal{L}$ forms a complete lattice under $\leq$.}
\end{quote}
 An earlier version was shown by Knaster and Tarski \cite{Knaster_Tarski}, who established the result for the special case where $\mathcal{L}$ is the lattice of subsets of a set (i.e. the power set lattice).
 
This is a classical theorem with broad applications. For example, in  formal semantics of programming languages and abstract interpretation,   the existence of fixed points can be exploited to guarantee well-defined semantics for a recursive algorithm \cite{tarski_pl_note}. In game theory, Tarski's theorem implies the existence of pure Nash equilibria in supermodular games~\cite{etessami2019tarski}.
Surprisingly, it is not fully understood how efficiently Tarski fixed points can be found. 

Formally, for $k,n \in \mathbb{N}$, let $\mathcal{L}_{n}^k = \{0, 1, \ldots, n-1\}^k$ be the $k$-dimensional grid of side length $n$.
Let $\leq$ be the binary relation where for vertices ${a} = (a_1, \ldots, a_k) \in \mathcal{L}_n^k$ and ${b} = (b_1, \ldots, b_k) \in \mathcal{L}_n^k$, we have  ${a} \leq {b}$ if and only if $a_i \leq b_i$ for each $i \in [k]$. 
We consider the lattice $(\mathcal{L}_n^k, \leq)$. A  function $f :\mathcal{L}_n^k \to \mathcal{L}_n^k$ is monotone if ${a} \leq {b}$ implies that $f({a}) \leq f({b})$. 
Tarski's theorem  states that the set $P$ of fixed points of $f$ is non-empty and that the system $(P, \leq)$ is itself a complete lattice \cite{tarski1955lattice}.
    
In this paper, we focus on the query model, where there is an unknown monotone function $f:\mathcal{L}_n^k \to \mathcal{L}_n^k$. An algorithm has to probe a vertex $v$ in order to learn the value of the function $f(v)$. The  task is to find a fixed point of $f$ by probing as few vertices as possible.  
The randomized query complexity  is the expected number of queries required to find a solution with a probability of at least $9/10$ \footnote{Any other constant greater than $1/2$ would suffice.}, where the expectation is taken over the coin tosses of the algorithm.

There  are two main  algorithmic approaches for finding a Tarski fixed point. The  first approach is a divide-and-conquer method that yields an upper bound of $O\left((\log{n})^{\lceil \frac{k+1}{2} \rceil}\right)$ for any fixed $k$ due to  \cite{chen2022improved}, which improves  an algorithm of  \cite{fearnley2022faster}. 

The second  is a path-following method that initially queries the vertex $\vec{0} = (0, \ldots, 0)$ and proceeds by following the directional output of the function.
With each function application, at least one coordinate is incremented, which guarantees that  a fixed point is reached within $O(nk)$ queries.

 \cite{etessami2019tarski} proved a randomized query complexity lower bound of $\Omega\left(\log^2(n)\right)$ on the 2D grid of side length $n$, which implies the same lower bound for the $k$-dimensional grid of side length $n$ when $k$ is constant. This lower bound shows that the divide-and-conquer algorithm is optimal for dimensions $k=2$ and $k=3$. 

For dimension $k \geq 4$, there is a growing gap between the best-known upper and lower bounds, since  
 the upper bound given by the divide-and-conquer algorithm has an exponential dependence on $k$. Meanwhile, the path-following method provides superior performance in high dimensions, such as the Boolean hypercube $\{0,1\}^k$, where it achieves an upper bound of $O(k)$.

\subsection{Our Contributions}

Let 
 ${TARSKI}(n,k)$ denote the Tarski search problem on the $k$-dimensional grid of side length $n$. 
\begin{definition} [${TARSKI}(n,k)$]
Let $k,n \in \mathbb{N}$.
     Given oracle access to an unknown monotone function $f : \mathcal{L}_n^k \to \mathcal{L}_n^k$, find a vertex $x \in \mathcal{L}_n^k$ with $f(x) = x$ using as few queries as possible.
 \end{definition}

 Note that the special case of the hypercube $TARSKI(2, k)$ is isomorphic to the power set lattice, as a $k$-dimensional bit vector can be interpreted as indicating which of $k$ elements are in a subset.

Our main result is the following:
\begin{theorem} \label{thm:intro_main}
The  randomized query complexity of ${TARSKI}(n,k)$ is  $\Omega \left( k + \frac{k \cdot \log{n}}{\log{k}} \right)$.
\end{theorem}
The lower bound in \cref{thm:intro_main} is sharp for   constant $n \geq 2$ and nearly optimal in the general case when $k$ is large relative to $n$.

\cref{thm:intro_main} gives a characterization of $\Theta(k)$ for the randomized  and deterministic query complexity on the Boolean hypercube $\{0,1\}^k$, and thus the power set lattice,  since the deterministic path-following method that iteratively applies the function starting from  vertex $\vec{0} = (0, \ldots, 0)$ finds a solution within $O(k)$ queries. 
No lower bound better than $\Omega(1)$ was known for the Boolean hypercube.

\begin{corollary}
    The randomized and deterministic query complexity of ${TARSKI}(2,k)$ is $\Theta(k)$.
\end{corollary}

We obtain  \cref{thm:intro_main} by designing the following family of monotone functions.




\begin{definition}[Set of functions $\mathcal{F}_n^k$] \label{def:F}
For each  $a \in \mathcal{L}_n^k$,  
we define a function $f^a : \mathcal{L}_n^k \to \mathcal{L}_n^k$ coordinate by coordinate. That is, 
for each $v =(v_1, \ldots, v_k) \in \mathcal{L}_n^k$ and $i \in [k]$, let
\begin{align} \label{eq:f_pi_q_def}
   f^a_i(v) = \begin{cases}
        v_i - 1 & \text{ if $(v_i > a_i)$ and $(v_j \le a_j$ for all $j<i)$}\\
        v_i + 1 & \text{ if $(v_i < a_i)$ and $(v_j \ge a_j$ for all $j<i)$}\\
        v_i & \text{ otherwise.}
    \end{cases}
\end{align}
Let $f^a(v) = (f^a_1(v), \ldots, f^a_k(v))$.
Define $\mathcal{F}_n^k = \{f^a \mid a \in \mathcal{L}_n^k\}$.
\end{definition}

The intuition is that the first digit that is too low and the first digit that is too high both get pushed towards their correct value.
An example of a function from Definition~\ref{def:F}  is shown in Figure~\ref{fig:example_function}. 

\begin{figure}[h!]
\centering 
\includegraphics[scale=0.6]{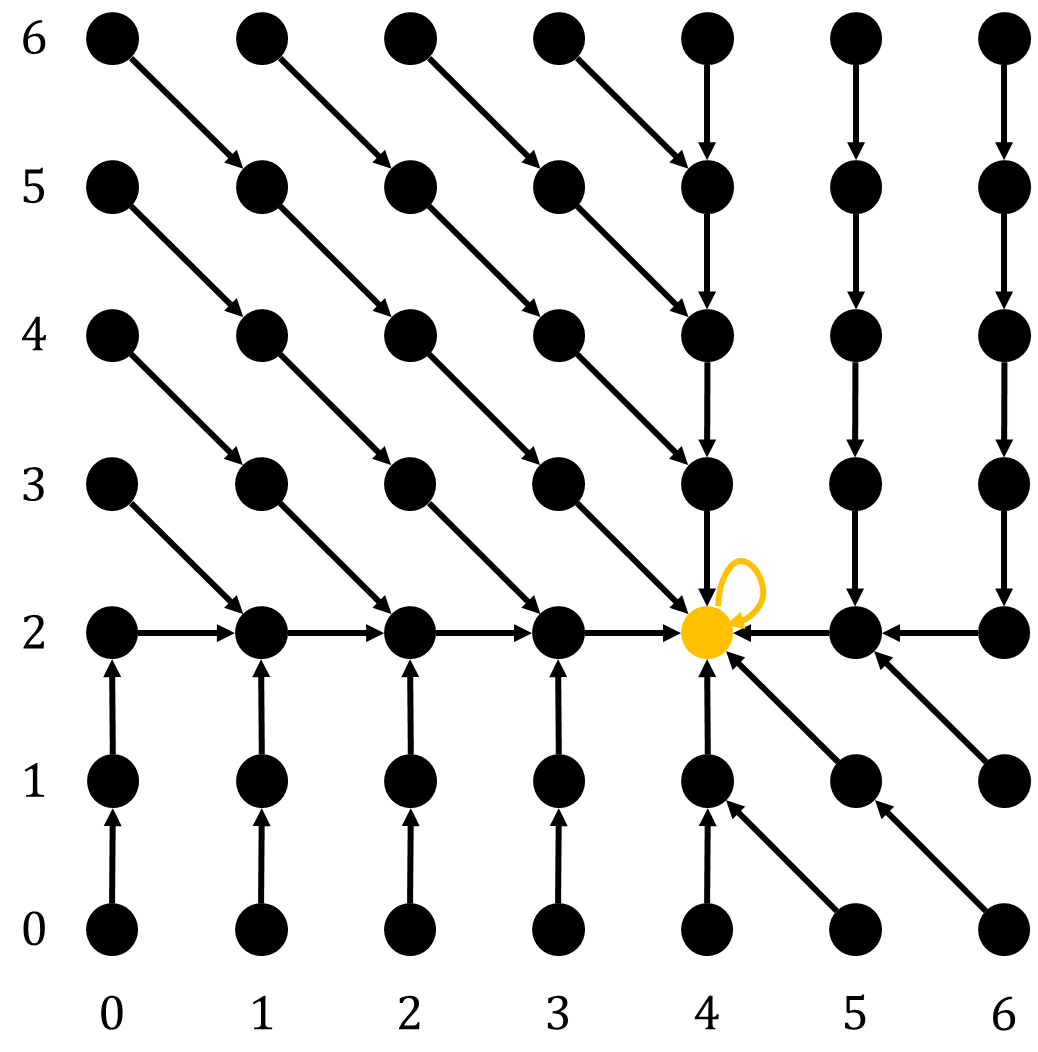}
\caption{Example of a function $f^a$ from Definition~\ref{def:F} where $a = (2,4)$ on the 2D grid of side length $7$ (i.e.  $n=7$ and $k=2$). The function $f^a$ has a unique fixed point at $a = (2,4)$. An arrow in the picture from a node $(u_1, v_1)$ to a node $(u_2, v_2)$ means that $f^a(u_1, v_1) = (u_2, v_2)$. For example, $f(0,0) = (0, 1)$. The fixed point is shown in yellow.}
\label{fig:example_function}
\end{figure}


For $k=2$, this construction is similar to the herringbone construction of \cite{etessami2019tarski}; however, our construction  does not induce a lower bound of  $\log^2(n)$ on the 2D grid since the shape of the path from $(0,0)$ or $(n-1,n-1)$ to the fixed point is too predictable. 

The real strength of our construction emerges  for large $k$, where the herringbone is not defined. Critically, a function $f \in \mathcal{F}_n^k$ has the property that $f(v)$ differs from $v$ in at most $2$ dimensions for all $v$, no matter how large $k$ is. This makes it difficult to derive information about more than a constant number of dimensions with a single query. The proof of \cref{thm:intro_main} makes this intuition precise. 

\subsection{Related Work}

\paragraph{Tarski fixed points.} Algorithms for the problem of finding Tarski fixed points on the $k$-dimensional grid of side length $n$ have only recently been considered. \cite{dang2020tarskialgorithm} gave an $O(\log^k (n))$ divide-and-conquer algorithm. \cite{fearnley2022faster} gave an $O(\log^2 (n))$ algorithm for the 3D grid and used it to construct an $O(\log^{2 \lceil k/3 \rceil} (n))$ algorithm for the $k$-dimensional grid of side length $n$. \cite{chen2022improved} extended their ideas to get an $O(\log^{\lceil (k+1)/2 \rceil}(n))$ algorithm.

 \cite{etessami2019tarski} showed a lower bound of $\Omega(\log^2(n))$ for the 2D grid, implying the same lower bound for the $k$-dimensional grid of side length $n$. This bound is tight for $k = 2$ and $k=3$, but there is an exponential gap for larger $k$. They also showed that the problem is in both PLS and PPAD, which by the results of \cite{fearnley2022cls} implies it is in CLS since PPAD $\cap$ PLS $=$ CLS.

\cite{CLY23} give a black-box reduction from the Tarski problem to the same problem with an additional promise that the input function has a unique fixed point. This result  implies that the Tarski problem and the unique Tarski problem have the same query complexity.

Next we briefly summarize  query and communication complexity results for two problems representative for the classes PLS and PPAD, respectively. These problems are finding  a  local minimum (representative for PLS) and a Brouwer fixed point of a continuous function  (representative for PPAD), respectively. In both cases, the existing lower bounds  also rely on hidden path constructions, which may be useful for proving  lower bounds in the Tarski setting. Query complexity lower bounds for PPAD $\cap$ PLS  were shown in \cite{HY17}. 

\paragraph{Brouwer fixed points.}	In the  Brouwer search problem, we are given a function $f : [0,1]^d \to [0,1]^d$ that is $L$-Lipschitz, for some constant $L > 1$. The algorithm has query access to the function $f$ and the task  is to find an $\eps$-approximate fixed point of $f$ using as few queries  as possible. The existence of a fixed point is guaranteed by Brouwer's fixed point theorem.

 The query complexity of computing an $\eps$-approximate Brouwer fixed point was studied in a series of papers starting with \cite{hirsch1989exponential},  which introduced a construction where  the function is induced by a  hidden walk. This was  later improved by \cite{chen2005algorithms} and \cite{chen2007paths}.

\paragraph{Local minima.} In the local search problem, we are given  a graph $G = (V,E)$ and a function $f : V \to \mathbb{R}$. A vertex $v$ is a local minimum if $f(v) \leq f(u) $ for all $(u,v) \in E$. An algorithm can probe a vertex $v$ to learn its value $f(v)$. The task is to find a vertex that is a local minimum using as few queries as possible. \cite{aldous1983minimization} obtains a lower bound  of $\Omega(2^{k/2-o(k)})$ on the query complexity for the Boolean hypercube $\{0,1\}^k$ by a random walk analysis.	
	Aldous' lower bound for the hypercube was later improved by 
	\cite{Aaronson06} to $\Omega(2^{k/2}/k^2)$ via a relational adversary method inspired from quantum computing.
	\cite{zhang2009tight} further improved this lower bound to  $\Theta(2^{k/2} \cdot \sqrt{k})$ via a ``clock''-based random walk construction.
	Meanwhile, \cite{llewellyn1989local} developed a deterministic divide-and-conquer algorithm.
	
	For the  $k$-dimensional grid $[n]^k$, 
	\cite{Aaronson06} used the relational adversary method  to show a randomized lower bound of $\Omega(n^{k/2-1} / \log n)$ for every constant $k \ge 3$. \cite{zhang2009tight} proved a randomized lower bound of $\Omega(n^{k/2})$ for every constant $k \ge 4$. 
	The work of \cite{sun2009quantum} closed further gaps in the quantum setting as well as the randomized $k=2$ case.

There are lower bounds for general graphs as a function of graph features such as separation number \cite{santha2004quantum} and vertex congestion \cite{BCR23}. There are upper bounds in terms of separation number \cite{santha2004quantum, llewellyn1989local} and genus~\cite{Verhoeven06}. \cite{babichenko2019communication} studied the communication complexity of local search, which this captures settings where data is stored  on different computers.

\section{Properties of the family of functions $\mathcal{F}_n^k$}

In this section we show that each function in the family $\mathcal{F}_n^k$ of Definition~\ref{def:F} is monotone and has a unique fixed point.

\begin{lemma}
For each $a \in \mathcal{L}_n^k$, the function $f^a$ from Definition \ref{def:F} is monotone.
\end{lemma}
\begin{proof}
    Consider two arbitrary vertices $u,v \in \mathcal{L}_n^k$ with $u \le v$.
    Suppose towards a contradiction that  $f^a(u) \le f^a(v)$ does not hold.
    Then there exists an index $i \in [k]$ such that $f^a_i(u) > f^a_i(v)$.
    Since $u \le v$, we have $u_i \le v_i$, so at least one of  $f^a_i(u) > u_i$ or  $f^a_i(v) < v_i$ holds.
    
    \begin{description} \item[{Case 1: $f^a_i(u) > u_i$.}]
        By definition of $f^a$, we then have $u_i < a_i$ and $u_j \ge a_j$ for all $j<i$.
        Since $u \le v$, we also have $v_j \ge a_j$ for all $j<i$.         Furthermore, we have  $v_i = u_i$, or $v_i = u_i+1$, or $v_i \ge u_i + 2$. We consider a few sub-cases:
        \begin{enumerate}[(a)]
        	\item ($v_i = u_i$). Then  $f^a_i(v) = v_i + 1 = u_i + 1 = f^a_i(u)$.
        	\item ($v_i = u_i+1$). Then  $v_i \le a_i$, so $f^a_i(v) \ge v_i = u_i + 1 = f^a_i(u)$.
        	\item ($v_i \ge u_i + 2$).
        Then $f^a_i(v) \ge v_i-1 \ge u_i+1 = f^a_i(u)$.
        \end{enumerate}
                In each subcase (a-c), we have $f^a_i(v) \ge f^a_i(u)$. This is in  contradiction with  $f^a_i(u) > f^a_i(v)$, thus case 1 cannot occur.

    \item[{Case 2: $f^a_i(v) < v_i$.}]
        By definition of $f^a$, we then have $v_i > a_i$ and $v_j \le a_j$ for all $j<i$.
        Since $u \le v$, we also have $u_j \le a_j$ for all $j<i$.
        Furthermore, we have  $u_i = v_i$, or $u_i = v_i - 1$, or $u_i \le v_i - 2$. We consider a few sub-cases:
        \begin{enumerate}[(a)]
        	\item ($u_i = v_i$). Then $f^a_i(u) = u_i - 1 = v_i - 1 = f^a_i(v)$.
        	\item ($u_i = v_i - 1$). Then $u_i \ge a_i$, so $f^a_i(u) \le u_i = v_i - 1 = f^a_i(v)$.
        	\item ($u_i \le v_i - 2$). Then $f^a_i(u) \le u_i + 1 \le v_i - 1 = f^a_i(v)$.       	
        \end{enumerate}
                In each subcase (a-c), we have $f^a_i(u) \le f^a_i(v)$. This in contradiction with $f^a_i(u) > f^a_i(v)$, thus case $2$ cannot occur either.
    \end{description}
    In both cases $1$ and $2$ we reached a contradiction, so the assumption that $f^a(u) \leq f^a(v)$  does not hold must have been false. Thus $f^a$ is monotone.
\end{proof}

\begin{lemma} \label{lem:f_unique_fixed_point}
For each $a \in \mathcal{L}_n^k$, the function $f^a \in \mathcal{F}_n^k$ has a unique fixed point at $a$.
\end{lemma}
\begin{proof}
By definition of $f^a$ we have $f^a(a) = a$, so $a$ is a fixed point of $f^a$.

Let $v \neq a$. 
    Then there exists $i \in [k]$ such that $v_i \ne a_i$. Let $i$ be the minimum such index. We have two cases:
    \begin{itemize}
\item ($v_i < a_i$):
        Then $f^a_i(v) = v_i + 1$, so $f^a(v) \neq v$.
\item ($v_i > a_i$):
        Then $f^a_i(v) = v_i - 1$, so $f^a(v) \neq v$.
\end{itemize}
    In both cases $v$ is not a fixed point, so $a$ is the only fixed point of $f^a$.
\end{proof}

\section{Lower bounds}

\subsection{Lower bound for the Boolean hypercube}
Using the family of functions $\mathcal{F}_n^k$ from \cref{def:F}, we can now  prove a randomized lower bound of $\Omega(k)$ for the Boolean hypercube $\{0,1\}^k$.

\begin{proposition}\label{thm:lb_k}
The randomized query complexity of ${TARSKI}(2,k)$ is $ \Omega(k)$.
\end{proposition}
\begin{proof}
We proceed by invoking Yao's lemma.
    Let  $\mathcal{U}$ be the uniform distribution over the set of functions $\mathcal{F}_2^k$. 
  Let $\mathcal{A}$ be the deterministic algorithm with the smallest possible expected number of queries that succeeds with probability at least $4/5$, where both the expected query count and the success probability are for input drawn from $\mathcal{U}$.
    The algorithm $\mathcal{A}$ exists since there is a finite number  of deterministic algorithms for this problem, so the minimum is well defined.
    
    Let $D$ be the expected number of queries issued by $\mathcal{A}$ on input drawn from $\mathcal{U}$.
    Let $R$ be the randomized query complexity of $TARSKI(2,k)$; i.e. the expected number of queries required to succeed with probability at least $9/10$.
    Then Yao's lemma (\cite{yao1977minimaxprinciple}, Theorem 3) yields $2R \ge D$. 
    Therefore it  suffices to lower bound $D$.  

Let $f^a \in \mathcal{F}_2^k$ be the function drawn from $\mathcal{U}$ on which $\mathcal{A}$ is run.
    For each  $t \in \mathbb{N}$, let
    \begin{itemize} 
    \item $\mathcal{H}_{t}^a$ denote the history of queries and responses received at steps  $1, \ldots, t$.
    \item $G_{2}^{k}(\mathcal{H}_t^a)$ denote set of all functions $f^b \in \mathcal{F}_2^k$ that are consistent with the given history $\mathcal{H}_t^a$ of query answers (meaning if those queries were conducted on $f^b$ one would get exactly the same answers as in $\mathcal{H}_t^a$).
    \item $\mathcal{I}_t^a = \Bigl\{i \in [k] \mid \mbox{For all } f^b \in G_2^k(\mathcal{H}_t^a), \mbox{ we have } b_i = a_i\Bigr\}$; that is, $\mathcal{I}_t^a$ represents the set of coordinates that the algorithm has learned with certainty after the first $t$ queries.
    \item $v^t$ denote the $t$-th query submitted by $\mathcal{A}$.
\end{itemize}


Given a set of indices $S \subseteq [k]$, a function $f^b \in \mathcal{F}_2^k$ is called \emph{consistent with $(S, a)$} if the unique fixed point $b$ of $f^b$ agrees with $a$ on all coordinates in $S$, in other words, if $b_i = a_i$ for all $i \in S$.

We claim that a function $f^b \in \mathcal{F}_2^k$ is consistent with $(\mathcal{I}_t^a, a)$ if and only if $f^b \in G_2^k(\mathcal{H}_t^a)$. In other words, at time $t$ the \emph{only} information $\mathcal{A}$ has is the value of $a_i$ for all $i \in \mathcal{I}_t^a$. 
The backwards direction of the claim is immediate: if $f^b \in G_2^k(\mathcal{H}_t^a)$, then by the definition of $\mathcal{I}_t^a$, we have $b_i = a_i$ for all $i \in \mathcal{I}_t^a$.

We prove the other direction by induction on $t$.
We have $G_2^k(\mathcal{H}_0^a) = \mathcal{F}_2^k$, so the base case of $t=0$  trivially holds.
We assume the inductive hypothesis holds for $t-1$ and prove it for $t$.

Consider an arbitrary function $f^b \in \mathcal{F}_2^k$ consistent with $(\mathcal{I}_t^a, a)$.
Since $\mathcal{I}_{t-1}^a \subseteq \mathcal{I}_t^a$, by the inductive hypothesis we have $f^b \in G_2^k(\mathcal{H}_{t-1}^a)$. In order to show that $f^b \in G_2^k(\mathcal{H}_t^a)$, we only need to show that $f^b(v^t) = f^a(v^t)$.

We consider the indices where $v^t$ has zeroes and divide in two cases; the analysis for indices where $v^t$ has ones is symmetric. Initialize $\mathcal{I}_t^a = \mathcal{I}_{t-1}^a$. We explain in each case what new indices may enter the set  $\mathcal{I}_t^a$.
\begin{enumerate}[(i)]
\item There exists $i \in [k]$ such that $\bigl( v^t_i = 0$ and $f^a_i(v^t) = 1 \bigr)$. This implies three things:
\begin{itemize}
    \item $a_i=1$, so $i$ is added to $\mathcal{I}_t^a$.
    \item For all $j < i$ such that $v^t_j = 0$, it must be the case that $a_j = 0$; otherwise, the bit at index $j$ would have been corrected to a $1$ instead of the bit at index $i$ getting corrected. Therefore, each such $j$ is added to $\mathcal{I}_t^a$.
    \item No information is revealed about the bits at locations  $j > i$ with $v^t_j = 0$, since regardless of the value of $a_j$, we have $f^a_j(v^t) = 0$. No such index $j$ is added to $\mathcal{I}_t^a$, though some may have already been in $\mathcal{I}_{t-1}^a$.
\end{itemize}

\item There is no $i \in [k]$ such that $\bigl( v^t_i = 0$ and $f^a_i(v^t) = 1 \bigr)$. Then for all $j \in [k]$ such that $v_j^t = 0$, it must have been the case that $a_j = 0$. Therefore, all such $j$ are added to $\mathcal{I}_t^a$.
\end{enumerate}

Suppose for contradiction that $f^b(v^t) \neq f^a(v^t)$. Let $i \in [k]$ be the smallest index where they differ, meaning $f^b_i(v^t) \neq f^a_i(v^t)$.
If $i \in \mathcal{I}_t^a$, then we would have $f^b_i(v^t) = f^a_i(v^t)$ since $f^b$ is consistent with $(\mathcal{I}_t^a, a)$; therefore, $i \notin \mathcal{I}_t^a$. We must further have $f^a_i(v^t) = v^t_i$, since otherwise $i$ would have been added to $\mathcal{I}_t^a$.

Without loss of generality, let $v^t_i = 0$ (the case where $v^t_i = 1 $ is symmetric). Since $i \notin \mathcal{I}_t^a$, there exists an index $j < i$ such that $v^t_j = 0$ and $f^a_j(v^t) = 1$. Because $j<i$ and $i$ is the smallest index for which $f_i^b(v^t) \neq f_i^a(v^t)$, we also have $f^b_j(v^t) = f^a_j(v^t) = 1$. But then $f^b_i(v^t) = 0 = f^a_i(v^t)$, as only one $0$ (the one at $j$) could have been changed to a $1$. This  contradicts the assumption that $f^b_i(v^t) \neq f^a_i(v^t)$. Thus we have $f^b(v^t) = f^a(v^t)$, which completes the inductive step.

We now have that a function $f \in \mathcal{F}_2^k$ is consistent with $(\mathcal{I}_t^a, a)$ if and only if $f \in G_2^k(\mathcal{H}_t^a)$. Suppose  algorithm $\mathcal{A}$ returned an answer after $t$ queries, where $|\mathcal{I}_t^a| < k$. Then there would be multiple functions consistent with $(\mathcal{I}_t^a, a)$, so we would have $|G_2^k(\mathcal{H}_t^a)| \geq 2$. Therefore, it has to guess the values of the coordinates it does not know, so it would make an error with probability at least $1/2$:
\begin{align} 
    \Pr \Bigl[\mbox{$\mathcal{A}$ succeeds within $t$ queries} \mid |\mathcal{I}_t^a| < k\Bigr] \leq \frac{1}{2} \,.
    \label{eq:A-fails-often-if-I-less-than-k}
\end{align}

Accordingly, we now seek to bound $|\mathcal{I}_t^a|$.
We argue that the expected number of bits learned with each query is upper bounded by a constant, that is $\Ex[|\mathcal{I}_t^a| - |\mathcal{I}_{t-1}^a| \mid \mathcal{H}_{t-1}^a] \leq 4$.

    


    For $\mathfrak{b} \in \{0, 1\}$, let $c_t(\mathfrak{b})$ be the index of the bit where $v_{c_t(\mathfrak{b})}^t = \mathfrak{b}$ and $f^a_{c_t(\mathfrak{b})}(v^t) = 1-\mathfrak{b}$, or $c_t(\mathfrak{b}) = \infty$ if no such index exists.
    There cannot be two such indices $c_t(\mathfrak{b})$ for any particular values of $t$ and $\mathfrak{b}$ since only the first digit that is lower (respectively higher) than the corresponding digit in $a$ is adjusted by functions $f^a \in \mathcal{F}_2^k$.
    Then define $\Delta_t(\mathfrak{b})$ as:
    \begin{equation}
        \Delta_t(\mathfrak{b}) = \left\{ i \in [k] \mid i \not \in \mathcal{I}_{t-1}^a, v_i^t = \mathfrak{b}, \mbox{ and } i \leq c_t(\mathfrak{b})\right\}\,.
    \end{equation}

    In other words $\Delta_t(\mathfrak{b})$ is precisely the set of indices $i \in [k]$ previously identified as being added to $\mathcal{I}_t^a$ (which were not in $\mathcal{I}_{t-1}^a$) with the property that $v_i^t = \mathfrak{b}$. Moreover, because $G_2^k(\mathcal{H}_t^a)$ is characterized only by $\mathcal{I}_t^a$, no additional indices are added. 
    We therefore have 
    \begin{align} \label{eq:delta_0_union_delta_1}
    \Delta_t(0) \cup \Delta_t(1) = \mathcal{I}_t^a \setminus \mathcal{I}_{t-1}^a \,.
    \end{align}

    An illustration of an execution history for the first three queries with the sets $\Delta_t(\mathfrak{b})$ can be found in Figure~\ref{fig:illustration_ct_deltat}.
    
\begin{figure}[h] 
\centering
\begin{tabular}{|c|c|c|c|c|c|c|c|}
\hline
$a$ & $0$ & $0$ & $1$ & $1$ & $1$ & $1$ & $0$ \\
\hline
$v^1$ & $0$ & $1$ & $1$ & $1$ & $0$ & $0$ & $1$ \\
\hline
$f^a(v^1)$ & \cellcolor[HTML]{ffbbbb}$0$ & \cellcolor[HTML]{bbbbff}$0$; $c_1(1)$ & $1$ & $1$ & \cellcolor[HTML]{ffbbbb}$1$; $c_1(0)$ & $0$ & $1$ \\
\hline
$v^2$ & \cellcolor[HTML]{cccccc}$0$ & \cellcolor[HTML]{cccccc}$0$ & $1$ & $0$ & \cellcolor[HTML]{cccccc}$1$ & $0$ & $1$ \\
\hline
$f^a(v^2)$ & \cellcolor[HTML]{cccccc}$0$ & \cellcolor[HTML]{cccccc}$0$ & \cellcolor[HTML]{bbbbff}$1$ & \cellcolor[HTML]{ffbbbb}$0$; $c_2(0)$ & \cellcolor[HTML]{cccccc}$1$ & $0$ & \cellcolor[HTML]{bbbbff}$0$; $c_2(1)$ \\
\hline
$v^3$ & \cellcolor[HTML]{cccccc}$0$ & \cellcolor[HTML]{cccccc}$0$ & \cellcolor[HTML]{cccccc}$1$ & \cellcolor[HTML]{cccccc}$1$ & \cellcolor[HTML]{cccccc}$1$ & $0$ & \cellcolor[HTML]{cccccc}$0$ \\
\hline
$f^a(v^3)$ & \cellcolor[HTML]{cccccc}$0$ & \cellcolor[HTML]{cccccc}$0$ & \cellcolor[HTML]{cccccc}$1$ & \cellcolor[HTML]{cccccc}$1$ & \cellcolor[HTML]{cccccc}$1$ & \cellcolor[HTML]{ffbbbb}$1$; $c_3(0)$ & \cellcolor[HTML]{cccccc}$0$ \\
\hline
\end{tabular}
\caption{An example series of three queries for $k=7$. Each row is  a vector in $\mathcal{L}_2^k$. The first row is the vector $a$, corresponding to the hidden fixed point. Each of the next rows represents either a query $v^t$ or its answer $f^a(v^t)$ (thus $v^1$, $f^a(v^1)$, $v^2$, $f^a(v^2)$, and so on). 
The red cells indicate the corresponding $\Delta_t(0)$, the blue cells indicate the set $\Delta_t(1)$, and the gray cells indicate the set $\mathcal{I}_t^a$. The finite values of $c_t(\mathfrak{b})$ are marked. For the third query, $c_3(1) = \infty$.}
\label{fig:illustration_ct_deltat}
\end{figure}

    The distribution of $|\Delta_t(\mathfrak{b})|$ can be bounded effectively. For any $C \in \mathbb{N}$, we have $|\Delta_t(\mathfrak{b})| > C$ only if the first $C$ indices $i \in [k]$ with the property that both  $i \not \in \mathcal{I}_{t-1}^a$ and  $v_i^t = \mathfrak{b}$ are guessed correctly, i.e. $a_i = \mathfrak{b}$. 
   Therefore:
    \begin{equation} \label{eq:bound_probability_delta_more_than_C}
        \Pr \bigl[|\Delta_t(\mathfrak{b})|  > C \mid \mathcal{H}_{t-1}^a \bigr] \leq 2^{-C} \,. 
    \end{equation}
    We can bound the expected value of $|\Delta_t(\mathfrak{b})|$ as:
    \begin{align}
        \mathbb{E}\bigl[|\Delta_t(\mathfrak{b})| \mid \mathcal{H}_{t-1}^a\bigr] &= \sum_{C=0}^{\infty} \Pr \bigl[|\Delta_t(\mathfrak{b})| > C \mid \mathcal{H}_{t-1}^a\bigr] \explain{By Lemma \ref{lem:natural-number-rv-expectation}} \\
        &\leq \sum_{C=0}^{\infty} 2^{-C} \explain{By  \eqref{eq:bound_probability_delta_more_than_C}}\\
        &= 2\,. \label{eq:upper_bound_size_of_delta_of_b}
    \end{align}

    Using \eqref{eq:delta_0_union_delta_1} and \eqref{eq:upper_bound_size_of_delta_of_b}, we get:
    \begin{align} \label{eq:difference_in_size_s_t_s_t_minus_1}
        \mathbb{E} \bigl[|\mathcal{I}_t^a| - |\mathcal{I}_{t-1}^a| \mid \mathcal{H}_{t-1}^a\bigr] &= \mathbb{E} \bigl[|\Delta_t(0)| \mid \mathcal{H}_{t-1}^a\bigr] + \mathbb{E} \bigl[|\Delta_t(1)| \mid \mathcal{H}_{t-1}^a\bigr] \leq 2 + 2 = 4\,. 
    \end{align}

    Since the upper bound of $4$ applies for all histories $\mathcal{H}_{t-1}^a$, taking expectation over all possible histories gives:
    \begin{equation} \label{eq:up_s_t_minus_s_t_minus_1_simplified}
        \mathbb{E} \bigl[|\mathcal{I}_t^a| - |\mathcal{I}_{t-1}^a|\bigr]  \leq 4 \qquad \forall t \in \mathbb{N}, t \geq 1 \,. 
    \end{equation}
    Since $|\mathcal{I}_0^a| = 0$, inequality \eqref{eq:up_s_t_minus_s_t_minus_1_simplified} implies  $\mathbb{E} \bigl[|\mathcal{I}_t^a|\bigr] \leq 4t$ for all $t \in \mathbb{N}$.
     
    Let $T = \lfloor  k/80 \rfloor $.
    Then by Markov's inequality applied to the random variable $\mathcal{I}_T^a$, we have
    \begin{align} \label{eq:prob-I-bigger-than-k-less-than-1-20}
        \Pr\bigl[|\mathcal{I}_T^a| \ge k\bigr] \le \frac{\E \bigl[|\mathcal{I}_T^a|\bigr]}{k} \leq 
        \frac{4T}{k} = \frac{4  \lfloor  k/80 \rfloor }{k} \leq \frac{1}{20} \,.
    \end{align}

When $\mathcal{A}$ succeeds within $T$ queries, it is because it learned all the coordinates (i.e. $|\mathcal{I}_t^a| \geq k$, which has probability at most $1/20$ by \eqref{eq:prob-I-bigger-than-k-less-than-1-20}) or it did not know all the coordinates by the end of the $T$-th query and guessed the remaining ones (meaning its success probability in this case would be at most $1/2$ by \eqref{eq:A-fails-often-if-I-less-than-k}). Combining these observations yields 
\begin{align} \label{eq:ub_success_with_few_queries}
    \Pr[\mathcal{A} \; \mbox{succeeds} \mid \mathcal{A} \; \mbox{ issued at most } T \mbox{ queries}] &\leq \frac{1}{20} + \frac{1}{2} =\frac{11}{20} \,.
\end{align}

Suppose for contradiction that the probability $\mathcal{A}$ makes more than $T$ queries is less than $1/4$. Then, by \eqref{eq:ub_success_with_few_queries}:

\begin{align}
    \Pr[\mathcal{A} \mbox{ succeeds}] &= \Pr[\mathcal{A} \; \mbox{succeeds} \mid \mathcal{A} \; \mbox{ issued at most } T \mbox{ queries}] \cdot \Pr[\mathcal{A} \mbox{ issued at most } T \mbox{ queries}]  \notag \\ 
    & \; 
    + \Pr[\mathcal{A} \; \mbox{succeeds} \mid \mathcal{A} \mbox{ issued more than}\; T \; \mbox{queries}] \cdot \Pr[ \mathcal{A} \mbox{ issued more than}\; T \; \mbox{queries}] \notag \\
    & < \frac{11}{20} \cdot 1 + 1 \cdot  \frac{1}{4}\notag \\
    & = \frac{4}{5} \,. \label{eq:contradiction_eq}
\end{align}
But $\mathcal{A}$ succeeds with probability at least $4/5$, which contradicts \eqref{eq:contradiction_eq}. Therefore, the expected number  of queries issued by $\mathcal{A}$ can be bounded as follows: 
\begin{align}
    D \geq T/4= (1/4) \cdot \lfloor k/80 \rfloor \in \Omega(k)\,.
\end{align}
This completes the proof. 
\end{proof}

We include   the statement of the next folk lemma, a proof of which can be  found, e.g.,  in \cite{MU_book}.
\begin{lemma}[Lemma 2.9 in \cite{MU_book}]
    \label{lem:natural-number-rv-expectation}
Let $X$ be a discrete random variable that takes  on only non-negative integer values. Then 
$        \Ex[X] = \sum_{i = 1}^{\infty} \Pr [X \geq i]\,.$
\end{lemma}

\subsection{Lower bound for the $k$-dimensional grid of side length $n$}
In this section we  show the randomized lower bound of $\Omega(k)$ also holds for $TARSKI(n,k)$. Afterwards, we prove the construction from \cref{def:F} also yields a lower bound of $\Omega\left(\frac{k \log{n}}{\log(k)}\right)$ for the $k$-dimensional grid of side length $n$.
\begin{lemma} \label{lem:grid_harder_than_hypercube}
 For $k, n \in \mathbb{N}$ with $n \geq 2$, the randomized query complexity of $TARSKI(n,k)$ is greater than or equal to 
   the randomized query complexity of $TARSKI(2,k)$.
\end{lemma}
\begin{proof}
    We show a reduction from $TARSKI(2,k)$ to $TARSKI(n,k)$.
    Let $f^* : \{0,1\}^k \to \{0,1\}^k$ be an arbitrary instance of $TARSKI(2,k)$. As such, $f^*$ is monotone.
    
    Let $g : \{0,1,\ldots,n-1\}^k \to \{0,1\}^k$ be the clamp function, given by 
    \begin{align} 
    g(v) = (g_1(v), \ldots, g_k(v)), \mbox{ where } g_i(v) = \min(v_i, 1)\; \forall i \in [k]\,.
    \end{align} 
    Then define $f : \{0,1,\ldots,n-1\}^k \to \{0,1,\ldots,n-1\}^k$ as $f(v) = f^*(g(v))$.

    To show that $f$ is monotone, let $u,v \in \{0,1,\ldots,n-1\}^k$ be arbitrary vertices with $u \le v$. Then $g(u) \le g(v)$, so $f(u) = f^*(g(u)) \le f^*(g(v)) = f(v)$ by monotonicity of $f^*$. Thus $f$ is monotone and has a fixed point.

    Every vertex $u \in \{0,1,\ldots,n-1\}^k \setminus \{0,1\}^k$ is not a fixed point of $f$, since $f(u) \in \{0,1\}^k$. Every fixed point $u \in \{0,1\}^k$ of $f$ is also a fixed point of $f^*$ since $g(u) = u$ for all $u \in \{0,1\}^k$. Therefore all fixed points of $f$ are also fixed points of $f^*$.

    A query to $f$ may be simulated using exactly one query to $f^*$ since computing $g$ does not require any knowledge of $f^*$.

    Therefore any algorithm that finds a fixed point of $f$ can also be used to find a fixed point of $f^*$ in the same number of queries. Therefore the randomized query complexity of $TARSKI(n,k)$ is greater than or equal to the randomized query complexity of $TARSKI(2,k)$.
 \end{proof}

Applying \cref{lem:grid_harder_than_hypercube} to \cref{thm:lb_k} directly gives the following corollary.
\begin{corollary} \label{eq:lb_n_k_as_corollary}
The randomized query complexity of ${TARSKI}(n,k)$ is $\Omega(k)$.
\end{corollary}

We also get the following lower bound for all $n$ and $k$.
 \begin{proposition}
    \label{thm:k-log-n-over-log-k-lower-bound}
    The randomized query complexity of  ${TARSKI}(n,k)$  is $\Omega\left( \frac{k\log(n)}{\log(k)}\right)$.
 \end{proposition}
 \begin{proof}
We invoke Yao's Lemma in exactly
the same way as in the proof of \cref{thm:lb_k}. That is, let  $\mathcal{U}$ be the uniform distribution over the set of functions $\mathcal{F}_n^k$. 
  Let $\mathcal{A}$ be the deterministic algorithm with the smallest possible expected number of queries that succeeds with probability at least $4/5$, where both the expected query count and the success probability are for inputs drawn from $\mathcal{U}$.
     $\mathcal{A}$ exists since here the number of deterministic algorithms is finite, so the minimum is well defined.
    
    Let $D$ be the expected number of queries issued by $\mathcal{A}$ on input drawn from $\mathcal{U}$.
    Let $R$ be the randomized query complexity of $TARSKI(n,k)$; i.e. the expected number of queries required to succeed with probability at least $9/10$.
    Then Yao's lemma (\cite{yao1977minimaxprinciple}, Theorem 3) yields $2R \ge D$. 
    Therefore it  suffices to lower bound $D$. 
    
    
     For each vertex $v \in \{0, \ldots, n-1\}^k$, let $Q_v$ be the set of possible outputs when plugging in $v$:
     \begin{align} 
     Q_v = \bigl\{f^a(v) \mid a \in \{0,1,\ldots,n-1\}^k\bigr\} \,.
     \end{align}
     We next bound $|Q_v|$.
     By the definition of $f^a$, the vertex $f^a(v)$ differs from $v$ in at most two coordinates: the first $i \in [k]$ such that $v_i > a_i$ (if any) and the first $j \in [k]$ such that $v_j < a_j$ (if any). Each of $i$ and $j$ have $k+1$ options, corresponding to the $k$ dimensions and the possibility that no such dimension exists. Therefore
     \begin{align}
         |Q_v| \le (k+1)^2 \,. \label{eq:Qv_bound}
     \end{align}
     
     Recall that $\mathcal{A}$ is defined to be the best deterministic algorithm  that succeeds on $\mathcal{U}$ with probability at least $4/5$.
     Since $\mathcal{U}$ is uniform over $\mathcal{F}_n^k$, there must exist at least $(4/5) \cdot n^k$ inputs on which $\mathcal{A}$ outputs a fixed point.
     Then the decision tree of $\mathcal{A}$ must have at least $(4/5) \cdot n^k$ leaves since all supported inputs have different and unique fixed points.
     Every node of this tree has at most $(k+1)^2$ children, since $|Q_v| \le (k+1)^2$ for all $v$ by \eqref{eq:Qv_bound}.
     Therefore the average depth of the leaves is at least
     \begin{align} \label{eq:D_bound}
         \log_{(k+1)^2}\Bigl((4/5)n^k \Bigr) - 1 = \frac{\log_2((4/5) n^k)}{\log_2((k+1)^2)} - 1 \ge
         \frac{k \log_2 (n) -1}{2\log_2 (k) + 2} - 1 \in \Omega\left(\frac{k \log n}{\log k}\right) \,.
     \end{align}


Then on input distribution $\mathcal{U}$, algorithm $\mathcal{A}$ issues an expected number of queries of $D \in \Omega(\frac{k \log n}{\log k})$.
     Thus the randomized query complexity of $TARSKI(n,k)$ is $\Omega(\frac{k \log n}{\log k})$ as required.
 \end{proof}

The proof of \cref{thm:intro_main} follows from \cref{eq:lb_n_k_as_corollary} and \cref{thm:k-log-n-over-log-k-lower-bound}.

\begin{proof}[Proof of \cref{thm:intro_main}]
The randomized query complexity of $TARSKI(n,k)$ is 
 $\Omega(k)$ by \cref{eq:lb_n_k_as_corollary} and 
 $\Omega\left(\frac{k \log{n}}{\log(k)}\right)$ by \cref{thm:k-log-n-over-log-k-lower-bound}.
This implies a lower bound of $\Omega\left( k + \frac{k \log{n}}{\log(k)}\right)$ as required.
\end{proof}

\section{Upper bounds for the family of functions $\mathcal{F}_n^k$}

Intuitively, the true query complexity of $TARSKI(n, k)$ on functions in $\mathcal{F}_n^k$ should be $\Theta(k \log n)$. After all, each query provides feedback on whether roughly two coordinates were too high or too low. 
This idea does give an $O(k \log n)$ upper bound, which we present in \cref{prop:k_logn_upper_bound}. However, this can be improved upon when $k$ is larger than $\frac{n}{\log{n}}$. We show this by providing  an $O(k+n)$ upper bound in \cref{prop:k_plus_n_upper_bound}. This implies that the family of functions $\mathcal{F}_n^k$ cannot give an $\Omega(k \log n)$ lower bound.

\begin{proposition} \label{prop:k_logn_upper_bound}
    There is a deterministic $O(k \log n)$-query algorithm for $TARSKI(n, k)$ on the set of functions $\mathcal{F}_n^k$.
\end{proposition}
\begin{proof}
Let $f^a$ be the hidden function. Our task is to learn $a \in \mathcal{L}_n^k$.
We consider the following algorithm that works in stages. By the end of  each stage $i$, the algorithm has learned all the values $a_1, \ldots, a_{i}$. 
\paragraph{Stage i.} Let ${u}_i^0 = 0$ and ${w}_i^0 = n-1$. The algorithm will do  binary search on the $i$-th coordinate by submitting queries of the form $(a_1, \ldots, a_{i-1}, z, 0, \ldots, 0)$ and recursing based on the answer. The invariant maintained is that at the $t$-th query in stage $i$, we have  $a_i \in [u_i^t, w_i^t]$. Formally, the $i$-th stage  works as follows.

\begin{enumerate}[(a)]
\item Initialize $t=0$.
\item While $u_i^t < w_i^t$:
\begin{itemize}
\item Query vertex  ${c}^t = \bigl(a_1, \ldots, a_{i-1}, \left\lfloor\left(u_i^t + w_i^t\right)/2\right\rfloor, 0, \ldots, 0\bigr)$ to obtain its value $f^a({c}^t)$. We denote by $c_i^t$ the $i$-th coordinate of the vector $c^t$. Consider a few cases:
\begin{enumerate}[(i)]
\item If $f_i^a(c^t) < c_i^t$ then:  $w_i^{t+1} = c_i^t-1$ and $u_i^{t+1} = u_i^t$.
\item If $f_i^a(c^t) > c_i^t$ then: $u_i^{t+1} = c_i^t+1$ and $w_i^{t+1} = w_i^t$.
\item If $f_i^a(c^t) = c_i^t$ then: $u_i^{t+1} = w_i^{t+1} = c_i^t$.
\end{enumerate}
\item Update $t = t+1$.
\end{itemize}
\item Now we have $a_i = u_i^t = w_i^t$.
\end{enumerate}
Then we move on to stage $i+1$ and return if all coordinates have been learned.

To show why the algorithm works, we argue that in each stage $i$, the recursion maintains the invariant $a_i \in [u_i^t, w_i^t]$ for all $t$. Assume for contradiction that this condition is first violated at some $i$ and $t$. It cannot be $t=0$, as $a_i \in [0, n-1] = [u_i^0, w_i^0]$. We consider each of the three cases (i)-(iii) that could have occurred in step $t-1$:
\begin{itemize}
    \item In case (i), we had $f_i^a(c^{t-1}) < c_i^{t-1}$. By the definition of $f^a_i$, this can only happen if $a_i < c_i^{t-1}$. Therefore, $a_i \leq c_i^{t-1}-1 = w_i^t$. Since $a_i \geq u_i^{t-1}$ and $u_i^t = u_i^{t-1}$, the invariant was preserved.
    \item In case (ii), we had $f_i^a(c^{t-1}) > c_i^{t-1}$. By the definition of $f^a_i$, this can only happen if $a_i > c_i^{t-1}$. Therefore, $a_i \geq c_i^{t-1}+1 = u_i^t$. Since $a_i \leq w_i^{t-1}$ and $w_i^t = w_i^{t-1}$, the invariant was preserved.
    \item In case (iii), we had $f_i^a(c^{t-1}) = c_i^{t-1}$. As this was supposedly the first violation of the invariant, all of $a_1, \ldots, a_{i-1}$ have been learned correctly by the algorithm. Therefore, $f_i^a(c^{t-1}) = c_i^{t-1}$ if and only if $c_i^{t-1} = a_i$. Accordingly, $u_i^t$ and $w_i^t$ are both set to $c_i^{t-1}$, preserving the invariant.
\end{itemize}
In all three cases, the supposed violation could not have occurred. This is a  contradiction, so the invariant always holds and the algorithm is correct.

Each step of the algorithm halves the gap between $u_i^t$ and $w_i^t$, so each stage only takes $O(\log n)$ queries. Since there are $k$ stages, the overall number of queries is $O(k \log n)$.
\end{proof}

Next we present an algorithm that gives an $O(k+n)$ upper bound.

\begin{proposition} \label{prop:k_plus_n_upper_bound}
    There is a deterministic $O(k + n)$-query algorithm for $TARSKI(n, k)$ on the set of functions $\mathcal{F}_n^k$.
\end{proposition}

\begin{proof}
    For each coordinate $i \in [k]$ and each query index $t$, let $x^t_i$ and $y^t_i$ be the minimum and maximum possible values of $a_i$ given the first $t$ queries the algorithm makes. For example, $x^0_i = 0$ and $y^0_i = n-1$ for all $i \in [k]$, and the algorithm finishes in $T$ queries if $x^T_i = y^T_i$ for all $i \in [k]$.

    If the algorithm has not finished by its $(t+1)$-st query, it queries the vertex with coordinates:
    \begin{equation*}
        v^{t+1}_i = \begin{cases}
            x^t_i & \text{if $x^t_i = y^t_i$} \\
            x^t_i + 1 & \text{otherwise}
        \end{cases}
    \end{equation*}

    There are two possible outcomes from each such query:
    \begin{itemize}
        \item Some coordinate $j$ satisfies $f^a_j(v^{t+1}) = v^{t+1}_j - 1$. This immediately identifies $a_j = x^t_j$, as $a_j \geq x^t_j$ and $a_j < x^t_j + 1$. Since there are only $k$ coordinates to learn, this case can occur for at most $k$ queries before the algorithm terminates.
        \item No coordinate $j$ satisfies $f^a_j(v^{t+1}) = v^{t+1}_j - 1$. Then, for each $i$ such that $x^t_i < y^t_i$, the possibility of $a_i = x^t_i$ is ruled out. Since each coordinate only has $n$ possible values, this case can occur for at most $n$ queries before the algorithm terminates.
    \end{itemize}

    Therefore, this algorithm terminates within $O(k+n)$ queries.
\end{proof}

\section{Discussion}

It would be interesting to characterize the query complexity of the Tarski search problem as a function of the grid side-length $n$ and dimension $k$.

\section{Acknowledgements}
We are grateful to Davin Choo and Kristoffer Arnsfelt Hansen for useful discussions.

\bibliographystyle{alpha}

\bibliography{tarski_bib}




\end{document}